\pdfoutput=1
\RequirePackage{ifpdf}
\ifpdf 
\documentclass[pdftex]{sigma}
\else
\documentclass{sigma}
\fi

\numberwithin{equation}{section}

\newtheorem{Theorem}{Theorem}[section]

\newcommand{\deriv}[3][]{\frac{d^{#1}{#2}}{{d{#3}}^{#1}}}
\newcommand{\pderiv}[3][]{\frac{\partial^{#1}{#2}}{{\partial{#3}}^{#1}}}
\def\dot#1{\deriv{#1}{\xi}}
\def\ddot#1{\deriv[2]{#1}{\xi}}
\def\bar#1{#1}
\def\d{d}\def\i{i}
\def\Ai{\mathop{\rm Ai}\nolimits}
\def\Bi{\mathop{\rm Bi}\nolimits}
\def\cdot{{\scriptstyle\,\bullet\,}}

\def\a{\alpha}
\def\b{\beta}
\def\ep{\varepsilon}
\def\la{\lambda}
\def\ga{\gamma}
\def\ph{\varphi}
\def\th{\vartheta}

\def\PI{${\rm P}_{\rm{I}}$}
\def\PII{${\rm P}_{\rm{II}}$}

\def\Ptf{$\mbox{\rm P}_{\rm XXXIV}$}

\def\N{\mathbb{N}}
\def\xx{\zeta}
\def\xii{s}
\def\zz{x}
\def\cI{c_1}
\def\cII{c_2}
\def\cIII{c_3}
\def\cIV{c_4}
\def\cV{c_5}
\def\cVI{c_6}
\def\cVII{c_7}

\begin{document}

\allowdisplaybreaks

\newcommand{\arXivNumber}{1701.03238}

\renewcommand{\PaperNumber}{018}

\FirstPageHeading

\ShortArticleName{Ermakov--Painlev\'{e}~II Symmetry Reduction of a Korteweg Capillarity System}

\ArticleName{Ermakov--Painlev\'{e}~II Symmetry Reduction\\ of a Korteweg Capillarity System}

\Author{Colin ROGERS~$^\dag$ and Peter A.~CLARKSON~$^\ddag$}

\AuthorNameForHeading{C.~Rogers and P.A.~Clarkson}

\Address{$^\dag$~Australian Research Council Centre of Excellence for Mathematics {\rm \&} Statistics\\
\hphantom{$^\dag$}~of Complex Systems, School of Mathematics, The University of New South Wales,\\
\hphantom{$^\dag$}~Sydney, NSW2052, Australia}
\EmailD{\href{mailto:c.rogers@unsw.edu.au}{c.rogers@unsw.edu.au}}

\Address{$^\ddag$~School of Mathematics, Statistics {\rm \&} Actuarial Science, University of Kent,\\
\hphantom{$^\ddag$}~Canterbury, CT2 7FS, UK}
\EmailD{\href{mailto:P.A.Clarkson@kent.ac.uk}{P.A.Clarkson@kent.ac.uk}}

\ArticleDates{Received January 13, 2017, in f\/inal form March 15, 2017; Published online March 22, 2017}

\Abstract{A class of nonlinear Schr\"odinger equations involving a triad of power law terms together with a de Broglie--Bohm potential is shown to admit symmetry reduction to a hybrid Ermakov--Painlev\'e~II equation which is linked, in turn, to the integrable Painlev\'e~XXXIV equation. A~nonlinear Schr\"odinger encapsulation of a Korteweg-type capillary system is thereby used in the isolation of such a Ermakov--Painlev\'e~II reduction valid for a multi-parameter class of free energy functions. Iterated application of a B\"acklund transformation then allows the construction of novel classes of exact solutions of the nonlinear capillarity system in terms of Yablonskii--Vorob'ev polynomials or classical Airy functions. A~Painlev\'e~XXXIV equation is derived for the density in the capillarity system and seen to correspond to the symmetry reduction of its Bernoulli integral of motion.}


\Classification{37J15; 37K10; 76B45; 76D45}

\section{Introduction}

Giannini and Joseph \cite{jgrj89}, in a nonlinear optics context, introduced a class of symmetry reductions for a cubic nonlinear Schr\"odinger (NLS) equation
\begin{gather}\label{eq:nls}
i\Psi_t + \Psi_{xx}+\nu |\Psi|^2\Psi=0,
\end{gather}
with subscripts denoted partial derivatives,
which resulted in the Painlev\'e II equation but with zero parameter $\a$
\begin{gather}\label{eq:PII0}
\deriv[2]{q}{z} = 2q^3+zq,
\end{gather}
see also \cite{ams15,ksc,ms12}; the reduction of the NLS equation \eqref{eq:nls} to equation \eqref{eq:PII0} was derived in \cite{gw89,taj83}. Numerical integration led to the isolation of interesting non-stationary solutions both bounded and stable in shape for a restricted range of ratio of nonlinearily to dispersion. However, the absence of the Painlev\'e parameter $\a$ in that reduction does not allow the iterative construction of sequences of exact solutions via the B\"acklund transformation for the canonical Painlev\'e~II~(\PII)
\begin{gather} \label{P2}
\deriv[2]{q}{z} = 2q^3 + z q + \a, \end{gather}
with $\a$ a parameter, as given by Gambier \cite{refGambier10} and Lukashevich \cite{nl71}. Here, symmetry reduction to \PII\ \eqref{P2} or to a hybrid Ermakov--Painlev\'e II equation linked to the integrable Painlev\'e XXXIV (\Ptf) equation
\begin{gather} \label{P34}
\deriv[2]{p}{z} = \frac{1}{2p} \left(\deriv{p}{z}\right)^2 + 2p^2 - z p - \frac{(\a+\tfrac12)^2}{2p}, \end{gather}
with $\a$ a parameter, is obtained for a wide class of NLS equations which incorporates a triad of power law terms together with a de Broglie--Bohm potential term. It is remarked that NLS equations involving such triple power law nonlinearities arise in nonlinear optics (see \cite{memm16, yxpsdmkkmmabmb16} and literature cited therein). Moreover, NLS equations containing a de Broglie--Bohm term also arise in the analysis of the propagation of optical beams \cite{crbmkcha10,wwhhjm68} as well as in cold plasma physics~\cite{jlopcrws07}. Under appropriate conditions, such `resonant' NLS equations admit novel fusion or f\/ission solitonic behaviour \cite{jlop07,jlopcrws07,opjl02,opjlcr08}.

The Ermakov--Painlev\'e II symmetry reduction is applied here to an NLS encapsulation of a~nonlinear capillarity system with origin in classical work of Korteweg~\cite{dk01}. Iterated application of a B\"acklund transformation admitted by~\PII~\eqref{P2} permits the construction via the linked \Ptf\ equation of novel multi-parameter wave packet solutions to the capillarity system in terms of either Yablonskii--Vorob'ev polynomials or classical Airy functions. These are shown to be valid for a multi-parameter class of model specif\/ic free energy relations. An invariance of the $(1+1)$-dimensional Korteweg capillarity system under a one-parameter class of reciprocal transformations as recently set down in \cite{crws14} allows the extension of the reduction procedure to a yet wider class of capillarity systems.

\section{A Ermakov--Painlev\'e II symmetry reduction}
Here, a class of $(1+1)$-dimensional nonlinear Schr\"odinger equations of the type
\begin{gather} \label{1}
\i \Psi_t + \Psi_{xx} - \left[ (1 - \mathcal{C}) \frac{|\Psi|_{xx}}{|\Psi|} -\i c \frac{|\Psi|_{x}}{|\Psi|^2} + \la |\Psi|^2 + \mu |\Psi|^{2m} + \nu |\Psi|^{2n} \right] \Psi = 0, \end{gather}
which incorporates a de Broglie--Bohm potential $|\Psi|_{xx}/|\Psi|$ and a triad of power law terms is investigated under a symmetry reduction. Thus, constraints on the parameters in \eqref{1} are sought for which the class of NLS equations admits symmetry reduction either to the \PII\ \eqref{P2}, with non-zero parameter $\a$, or to a hybrid Ermakov--Painlev\'e II equation under a wave packet ansatz
\begin{subequations}\label{2-3}\begin{gather} \label{2}
\Psi = [ \phi (\xi) + \i \psi (\xi) ] \exp(\i \eta), \end{gather}
with
\begin{gather} \label{3}
\xi = \bar{\a} t + \bar{\b} t^2 + \bar{\ga} x ,\qquad \eta = \bar{\gamma} t^3 + \bar{\delta} t^2 + \bar{\ep} \bar{\ga} t x + \bar{\zeta} t + \bar{\la} x.\end{gather}\end{subequations}
In the nonlinear optics context of \cite{jgrj89} such a similarity transformation was used to reduce a~standard cubic NLS equation to Painlev\'e~II but with zero parameter $\a$ and resort was made to a numerical treatment. Asymptotic properties of \PII\ \eqref{eq:PII0} with $\a=0$ have been discussed by various authors, see, e.g., \cite{mahs81,bclm,cmcl88,dz95,fikn06,hmcl81,jm78}.

In the present case, on introduction of the wave packet ansatz \eqref{2-3} into \eqref{1}, it is seen that
\begin{subequations} \label{4-5}\begin{gather} \label{4}
\bar{\ga}^2 \ddot{\phi} - \dot{\psi} \big[ 2\big(\bar{\b} + \bar{\ep} \bar{\ga}^2\big) t + \bar{\a} + 2\bar{\la} \bar{\ga} \big] + \frac{c\bar{\gamma}\psi}{|\Psi|^3} \left(\phi \dot{\phi} + \psi \dot{\psi}\right) - \Delta \phi = 0 , \\ 
\label{5}
\bar{\ga}^2 \ddot{\psi} + \dot{\phi} \big[ 2\big(\bar{\b} + \bar{\ep} \bar{\ga}^2\big) t + \bar{\a} + 2\bar{\la} \bar{\ga} \big] - \frac{c\bar{\gamma} \phi}{|\Psi|^3} \left(\phi \dot{\phi} + \psi \dot{\psi}\right) - \Delta \psi = 0, \end{gather}\end{subequations}
where
\begin{gather}
\Delta = 3\bar{\ga} t^2 + 2 \bar{\delta} t + \bar{\ep} \bar{\ga} x + \bar{\zeta} + (\bar{\ep} \bar{\ga} t + \bar{\la})^2 + \la |\Psi|^2 + \mu |\Psi|^{2m} + \nu |\Psi|^{2n} \nonumber\\
\hphantom{\Delta =}{} + \dfrac{s \bar{\ga}^2}{|\Psi|^4} \left\{ \left[\phi \ddot{\phi} + \psi \ddot{\psi} + \left(\dot{\phi}\right)^2 + \left(\dot{\psi}\right)^2\right]
\left(\phi^2 + \psi^2\right) - \left(\phi \dot{\phi} + \psi \dot{\psi}\right)^2 \right\},\!\!\! \label{6}
\end{gather}
with $s=1-\mathcal{C}$. The relations \eqref{4-5} together show that
\begin{gather}
\bar{\ga}^2 \left(\ddot{\phi} \psi - \ddot{\psi} \phi\right) - \left(\phi \dot{\phi} + \psi \dot{\psi}\right) \big[ 2\big(\bar{\b} + \bar{\ep} \bar{\ga}^2\big) t + \bar{\a} + 2\bar{\la} \bar{\ga} \big]\nonumber\\
\qquad{} + \frac{c\bar{\gamma}}{|\Psi|} \left(\phi \dot{\phi} + \psi \dot{\psi}\right) = 0, \label{7} \end{gather}
whence it is required that
\begin{gather*} \label{8}
\bar{\b} + \ep \bar{\ga}^2 = 0, \end{gather*}
in which case equation \eqref{7} admits the integral
\begin{gather} \label{9}
\bar{\ga}^2 \left(\dot{\phi} \psi - \dot{\psi} \phi\right) - \tfrac{1}{2} (\bar{\a} + 2\bar{\la} \bar{\ga}) |\Psi|^2 + c \bar{\ga} |\Psi| = \mathcal{I}, \end{gather}
where $\mathcal{I}$ is an arbitrary constant of motion.

On use of the relation
\begin{gather} \label{10}
\left[\phi \ddot{\phi} + \psi \ddot{\psi} + \left(\dot{\phi}\right)^2 + \left(\dot{\psi}\right)^2\right] (\phi^2 + \psi^2) - \left(\phi \dot{\phi} + \psi \dot{\psi}\right)^2 =|\Psi|^3 \deriv[2]{|\Psi|}{\xi} , \end{gather}
it is seen that \eqref{6} yields, if $\bar{\b}\neq0$,
\begin{gather}
 \Delta = \bar{\ga} \bar{\b}^{-1} \big(\bar{\ep}^2 \bar{\ga} + 3\big) (\xi - \bar{\a} t - \bar{\ga} x) + 2(\bar{\delta} + \bar{\ep} \bar{\ga} \bar{\la}) t + \bar{\ep} \bar{\ga} x + \bar{\zeta} + \bar{\la}^2 \nonumber\\
\hphantom{\Delta =}{} + \la |\Psi|^2 + \mu |\Psi|^{2m} + \nu |\Psi|^{2n} + \dfrac{s \bar{\ga}^2}{|\Psi|} \deriv[2]{|\Psi|}{\xi} \nonumber\\
\hphantom{\Delta}{}= \bar{\ep} \xi + \bar{\zeta} + \bar{\la}^2 + \la |\Psi|^2 + \mu |\Psi|^{2m} + \nu |\Psi|^{2n} + \dfrac{s \bar{\ga}^2}{|\Psi|} \deriv[2]{|\Psi|}{\xi},\label{11}\end{gather}
on setting
\begin{gather} 
\bar{\b} \bar{\ep} = \bar{\ga} \big(3 + \bar{\ep}^2 \bar{\ga}\big),\qquad 
\bar{\a} \bar{\ep} = 2(\bar{\delta} + \bar{\ep} \bar{\ga} \bar{\la}). \label{13}\end{gather}
Moreover, equations \eqref{4-5} again combine to show that
\begin{gather*} \label{14}
\bar{\ga}^2 \left(\phi \ddot{\phi} + \psi \ddot{\psi}\right) + \left(\dot{\phi} \psi - \dot{\psi} \phi\right) (\bar{\a} + 2 \bar{\la} \bar{\ga})
- \Delta |\Psi|^2 = 0, \end{gather*}
whence, on use of the identity
\begin{gather*} \label{15}
(\phi^2 + \psi^2) \left[\left(\dot{\phi}\right)^2 + \left(\dot{\psi}\right)^2\right] - \left(\dot{\phi} \psi - \dot{\psi} \phi\right)^2 \equiv \left(\phi \dot{\phi} + \psi \dot{\psi}\right)^2, \end{gather*}
together with \eqref{10}, it is seen that
\begin{gather*} \label{16}
\bar{\ga}^2 \left[ |\Psi|^3 \deriv[2]{|\Psi|}{\xi}- \left(\dot{\phi} \psi - \dot{\psi} \phi\right)^2 \right] + (\bar{\a} + 2 \bar{\la} \bar{\ga}) \left(\dot{\phi} \psi - \dot{\psi} \phi\right) |\Psi|^2 - \Delta |\Psi|^4 = 0.\end{gather*}
The latter, by virtue of the integral of motion \eqref{9} and the expression \eqref{11} for $\Delta$ now produces a nonlinear equation in the amplitude $ |\Psi|$, namely
\begin{gather}
\deriv[2]{|\Psi|}{\xi} + [ {{\cI}} + {{\cII}} \xi ] |\Psi| + {{\cIII}} |\Psi|^3 + {{\cIV}} |\Psi|^{2m+1} + {{\cV}} |\Psi|^{2n+1} + \frac{{{\cVI}} }{|\Psi|} + \frac{{{\cVII}}}{ |\Psi|^2}\nonumber\\
\qquad{} = \frac{\mathcal{I}^2} {(1-s)\bar{\ga}^4 |\Psi|^3}, \label{17}\end{gather}
where the constants ${{\cI}},{{\cII}},\ldots,{{\cVII}}$ are given by
\begin{subequations}\label{18}\begin{gather}
{{\cI}} = \dfrac{\left( \bar{\a} - {\bar{\delta}}/{\bar{\ep}} \right)^2 - \bar{\ga}^2 (\bar{\zeta} + \bar{\la}^2)} {(1-s)\bar{\ga}^4}, 
\qquad {{\cII}} = \dfrac{\bar{\ep}} {(s-1)\bar{\ga}^2} ,\qquad {{\cIII}} = \dfrac{\la} {(s-1)\bar{\ga}^2} ,\\
{{\cIV}} = \dfrac{\mu} {(s-1)\bar{\ga}^2} , \qquad 
{{\cV}} = \dfrac{\nu} {(s-1)\bar{\ga}^2} ,\qquad {{\cVI}} = \dfrac{c^2} {(s-1)\bar{\ga}^2} ,\qquad {{\cVII}} = \dfrac{2c
\mathcal{I}} {(1-s)\bar{\ga}^3} ,
 \end{gather}\end{subequations}
and it is required that $s\neq1$.

Below a triad of cases is set down in which the amplitude equation \eqref{17} reduces either directly to \PII\ or to a hybrid Ermakov--Painlev\'e II equation subsequently shown to be integrable.

\subsubsection*{Case (i) $\boldsymbol{\mathcal{I}=0}$; $\boldsymbol{m=-\tfrac12}$; $\boldsymbol{n=-1}$.}
In this case with
\begin{gather*} \label{19}
{{\cI}} = 0 ,\qquad {{\cII}} = -1 ,\qquad {{\cIII}} = -2 ,\qquad {{\cIV}} = -\a ,\qquad {{\cV}} + {{\cVI}} = 0,\end{gather*}
the amplitude equation reduces directly to the Painlev\'e II equation
\begin{gather} \label{20}
\deriv[2]{|\Psi|}{\xi} = 2 |\Psi|^3 + \xi |\Psi| + \a , \end{gather}
corresponding to the symmetry reduction via the ansatz \eqref{2-3} of the class of NLS equations
\begin{gather*} \label{21}
\i \Psi_t + \Psi_{xx} - \left[ (1 - \mathcal{C}) \frac{|\Psi|_{xx}}{|\Psi|} -\i c \frac{|\Psi|_{x}}{|\Psi|^2} + \mathcal{C} \bar{\ga}^2 |\Psi|^2 + \frac{\mathcal{C} \bar{\ga}^2 \a}{ |\Psi|} - \frac{c^2}{ |\Psi|^2} \right] \Psi = 0. \end{gather*}

\subsubsection*{Case (ii) $\boldsymbol{\mathcal{I}=0}$; $\boldsymbol{c=0}$; $\boldsymbol{m=-\tfrac12}$; $\boldsymbol{n=0}$.}
Here, with
\begin{gather*} \label{22}
{{\cI}} = -{{\cV}} ,\qquad {{\cII}} = -1 ,\qquad {{\cIII}} = -2 ,\qquad {{\cIV}} = -\a, \end{gather*}
the \PII\ equation \eqref{20} again results, while the associated class of NLS equations \eqref{1} becomes
\begin{gather} \label{23}
\i \Psi_t + \Psi_{xx} - \left[ (1 - \mathcal{C}) \frac{|\Psi|_{xx}}{|\Psi|} + \mathcal{C} \bar{\ga}^2 |\Psi|^2 + \frac{\mathcal{C} \bar{\ga}^2 \a}{ |\Psi|} + \nu \right] \Psi = 0.\end{gather}
It is remarked that in the absence of the de Broglie--Bohm term, a time-independent NLS equation of this type~\eqref{23} incorporating a nonlinearity $\sim|\Psi|^{-1}$ has been derived `ab initio' in~\cite{ysjh94} via a geometric model which describes stationary states of supercoiled DNA.

\subsubsection*{Case (iii) $\boldsymbol{\mathcal{I}\neq0}$; $\boldsymbol{c=0}$; $\boldsymbol{m=-\tfrac12}$; $\boldsymbol{n=0}$.}
In this case, equation \eqref{17} reduces to a hybrid `Ermakov--Painlev\'e II' equation of the type
\begin{gather} \label{24}
\deriv[2]{|\Psi|}{\xi} + \ep |\Psi|^3 + (\delta \xi + \zeta) |\Psi| = \frac{\sigma}{|\Psi|^{3}}, \end{gather}
and which will be subsequently seen to be linked to \Ptf\ \eqref{P34}. It is recalled that the classical Ermakov equation with roots in \cite{ve80}, namely
\begin{gather*} \label{25}
\deriv[2]{\bar{\mathcal{E}}}{\xi} + \omega(\xi) \bar{\mathcal{E}} = \frac{\sigma}{\bar{\mathcal{E}}^{3}}, \end{gather*}
admits a nonlinear superposition principle readily derived via a Lie group approach as in \cite{crur89,crwspw97}.

In the subsequent application to the Korteweg capillarity system it will be the Ermakov--Painlev\'e II symmetry reduction that will be exploited. With a positive solution $|\Psi|$ of \eqref{24} to hand, the corresponding class of exact solutions for $\Psi$ in the wave packet representation~\eqref{2-3} is obtained via the integral of motion \eqref{9}. Thus, the latter yields
\begin{gather*} 
\bar{\ga}^2 \deriv{}{\xi} \left[ \tan^{-1} \left( \frac{\phi}{\psi} \right) \right] - \bar{\a} + \frac{\bar{\delta}}{\bar{\ep}} + \frac{c \bar{\ga}}{|\Psi|} = \frac{\mathcal{I}}{|\Psi|^2}, \end{gather*}
whence, on integration
\begin{gather} \label{27}
\bar{\ga}^2 \tan^{-1} \left( \frac{\phi}{\psi} \right) = \left( \bar{\a} - \frac{\bar{\delta}}{\bar{\ep}} \right) \xi - c \bar{\ga} \int \frac{1}{|\Psi|}\,\d \xi + \mathcal{I} \int \frac{1}{|\Psi|^2}\,\d \xi, \end{gather}
where use has been made of the relation \eqref{13}. Accordingly, with $V=\phi/\psi$, it is seen that~$\phi$,~$\psi$ in the original wave packet representation are given by the relations
\begin{gather*} 
\phi = \pm\frac{ |\Psi| V }{ \sqrt{1 + V^2}},\qquad \psi = \pm \frac{|\Psi| }{\sqrt{1 + V^2}}.\end{gather*}

In the sequel, the link between the Ermakov--Painlev\'e~II equation \eqref{24} and \Ptf\ \eqref{P34} is used to construct novel classes of wave packet solutions of a nonlinear Korteweg capillarity system in terms of Yablonskii--Vorob'ev polynomials or classical Airy functions via the iterated application of the B\"acklund transformation for \PII\ due to Gambier~\cite{refGambier10} and Lukashevich~\cite{nl71}.

\section{The capillarity system}

In \cite{la96}, Antanovskii derived the isothermal capillarity system with continuity equation
\begin{subequations}\label{29-30}\begin{gather} \label{29}
\rho_t + \operatorname{div} (\rho \mathbf{v}) = 0 , \end{gather}
augmented by the momentum equation
\begin{gather} \label{30}
\mathbf{v}_t + \mathbf{v}\cdot\nabla \mathbf{v} + \nabla \left[ \frac{\delta (\rho \mathcal{E})}{\delta \rho} - \Pi \right] = \mathbf{0} ,
\end{gather}\end{subequations}
where $\rho$ is the density, $\mathbf{v}$ velocity and $\mathcal{E}(\rho,|\nabla\rho|^2/\rho)$ is the specif\/ic free energy. Herein, the standard variational derivative notation
\begin{gather*} 
\dfrac{\delta \Theta}{\delta \rho} = \dfrac{\partial \Theta}{\partial \rho} - \nabla \cdot \left[ \dfrac{\partial \Theta}{\partial \mathcal{A}} \nabla \rho \right], \end{gather*}
is adopted with $\mathcal{A}=\tfrac12|\nabla\rho|^2$. In the above $\Pi$ is an external potential, commonly taken to be that due to gravity, in which case $\Pi=-\rho g$. The classical Korteweg capillarity system as set down in \cite{dk01} is retrieved in the specialisation
\begin{gather*} 
\mathcal{E} (\mathcal{A}, \rho) = \frac{\kappa(\rho) \mathcal{A}}{\rho} ,\qquad \kappa(\rho) > 0, \end{gather*}
in which case, the momentum equation becomes
\begin{gather*} 
\mathbf{v}_t + \mathbf{v}\cdot\nabla \mathbf{v} - \nabla \left[ \kappa(\rho) \nabla^2 \rho + \tfrac12 |\nabla \rho|^2 \deriv{\kappa(\rho)}{\rho} + \Pi \right] = 0.\end{gather*}
The classical Boussinesq capillarity system, in turn, is retrieved as the specialisation with $\kappa$ constant in this Korteweg system. A system analogous to the Boussinesq model arises `mutatis mutandis' in plasma physics~\cite{vbsv05}.

In the case of irrotationality with $\mathbf{v}=\nabla\Phi$, the momentum equation \eqref{30} admits the Bernoulli integral
\begin{gather*} 
\Phi_t + \tfrac12 |\nabla\Phi|^2 + \frac{\delta}{\delta\rho} (\rho \mathcal{E}) - \Pi = \mathcal{B}(t), \end{gather*}
and on introduction of the Madelung transformation \cite{em26}
\begin{gather} \label{35}
\Psi = \rho^{1/2} \exp \left( \tfrac12{i \Phi} \right), \end{gather}
the capillarity system \eqref{29-30} may be encapsulated in the generalised NLS-type equation
\begin{gather} \label{36}
\i \Psi_t + \nabla^2 \Psi + \left[ - \frac{\nabla^2 |\Psi|}{|\Psi|} - \tfrac12 \frac{\delta (\rho \mathcal{E})}{\delta \rho} + \tfrac12{\Pi} \right] \Psi = 0, \end{gather}
incorporating a de Broglie--Bohm potential term.

It was observed by Antanovskii et al.\ in \cite{lacrws87} that if $\Pi=0$ and
\begin{gather} \label{37}
\mathcal{E} \left( \tfrac12 |\nabla \rho|^2, \rho \right) = \mathcal{C} \frac{ |\nabla \rho|^2}{2\rho^2} + \nu \rho + \frac{\tau}{\rho}, \end{gather}
then \eqref{36} reduces, if $\mathcal{C}=1$, to the cubic nonlinear Schr\"odinger equation
\begin{gather*} \label{38}
\i \Psi_t + \nabla^2 \Psi - \nu |\Psi|^2 \Psi = 0.\end{gather*}
If, on the other hand, $\mathcal{C}\neq1$, it is seen that reduction is obtained to a `resonant' NLS-type equation \cite{crws99}
\begin{gather} \label{39}
\i \Psi_t + \nabla^2 \Psi + \left[(\mathcal{C}-1) \frac{\nabla^2 |\Psi|}{|\Psi|} - \nu |\Psi|^2 \right] \Psi = 0.\end{gather}
Moreover, if $\mathcal{C}>0$ as in the present capillarity context, \eqref{39} may be transformed to a standard cubic NLS equation with the de Broglie--Bohm term removed (see, e.g.,~\cite{cr14}). Thus, in $1+1$ dimensions with three-parameter model energy $\mathcal{E}(\tfrac12|\nabla\rho|^2,\rho)$ of the type \eqref{37} reduction is made to a canonical integrable NLS equation. The capillarity system encapsulated in the $(1+1)$-dimensional version of~\eqref{39} then becomes amenable to established methods of soliton theory such as inverse scattering procedures and inherits admittance of invariance under a B\"acklund transformation together with concomitant nonlinear superposition principle (see, e.g., \cite{mapc91,crws98,crws02} and literature cited therein). It is noted that a gravitational potential term $\Pi=-\rho g$ is readily accommodated in the above reduction. Detailed qualitative properties of capillarity systems with model laws of the type~\eqref{37} with K\'{a}rm\'{a}n--Tsien-type law
\begin{gather} \label{40}
\kappa(\rho) = {\mathcal{C}}/{\rho} ,\qquad \mathcal{C} > 0, \end{gather}
have been recently set down in \cite{rcrdjs12} while travelling wave propagation in $(1+1)$-dimensional capillarity theory has been investigated in~\cite{sg13}.

Here, a more general class of model energy $\mathcal{E}(\tfrac12|\nabla\rho|^2,\rho)$ laws is considered, namely that with
\begin{gather*} 
\mathcal{E} = \frac{\kappa(\rho)|\nabla \rho|^2}{2 \rho} + \frac{\mathcal{R}(\rho)}{\rho}, \end{gather*}
so that, with $\Pi=0$, \eqref{36} produces the class of NLS equations
\begin{gather} \label{42}
\i \Psi_t + \nabla^2 \Psi - \left[\left(1 + \deriv{\kappa}{\rho} |\Psi|^4\right) \frac{\nabla^2 |\Psi|}{|\Psi|} - \tfrac12 \left(\kappa(\rho) + \deriv{\kappa}{\rho} |\Psi|^2\right) \left(\nabla |\Psi|\right)^2 + \tfrac12 \deriv{\mathcal{R}}{\rho} \right] \Psi = 0. \end{gather}
Thus, capillarity systems encapsulated in \eqref{42} are isolated which may be aligned with NLS equations of the type \eqref{1} in the case $c=0$, $m=-2$, $n=0$. This occurs for the multi-parameter class of model energy laws with
\begin{gather} \label{43}
\mathcal{E} \big( \tfrac12|\nabla \rho|^2 , \rho \big) = \mathcal{C} \frac{ |\nabla \rho|^2} {2 \rho^2} + \la \rho - \frac{2 \mu}{\rho^2} + 2 \nu + \frac{\tau}{\rho}, \end{gather}
where $\la$, $\mu$, $\nu$ and $\tau$ together with $\mathcal{C}>0$ are real constants. Importantly, this includes in the case $\mu=0$ the class which has been recently subject to a detailed qualitative analysis in \cite{rcrdjs12}. The $\kappa(\rho)$ capillarity relation is seen to be of the K\'{a}rm\'{a}n--Tsien type \eqref{40}. Here,
\begin{gather} \label{44}
\la = - \mathcal{C} \bar{\ga}^2 {{\cIII}} ,\qquad \mu = - \mathcal{C} \bar{\ga}^2 {{\cIV}} ,\qquad \nu = - \mathcal{C} \bar{\ga}^2 {{\cV}}, \end{gather}
in accordance with the relations \eqref{18}. The associated class of NLS equations
\begin{gather*} 
\i \Psi_t + \Psi_{xx} - \left[ (\mathcal{C}-1) \frac{|\Psi|_{xx}}{|\Psi|} + \la |\Psi|^2 + \frac{\mu}{|\Psi|^4} + \nu \right] \Psi = 0, \end{gather*}
hence, admits symmetry reduction via the wave packet ansatz \eqref{2-3} to the hybrid Ermakov--Painlev\'e II equation (cf.~\eqref{17})
\begin{gather} \label{46}
\deriv[2]{|\Psi|}{\xi} + \left[ {{\cI}} + {{\cV}} + {{\cII}} \xi \right] |\Psi| + {{\cIII}} |\Psi|^3 = \frac{\sigma}{|\Psi|^{3}}, \end{gather}
where
\begin{gather} \label{47}
\sigma = \frac{1}{\mathcal{C} \bar{\ga}^2} \left[ \left( \frac{\mathcal{I}}{ \bar{\ga}} \right)^2 + \mu \right].\end{gather}
Interestingly, this symmetry reduction to an integrable Ermakov--Painlev\'e II equation will be admitted by Korteweg-type capillarity systems with the particular model energy laws of the type discussed in~\cite{rcrdjs12}.

Under the translation $\xx=\xi+{({{\cI}}+{{\cV}})}/{{{\cII}}}$, with ${{\cII}}\neq0$, \eqref{46} becomes
\begin{gather} \label{48}
\deriv[2]{|\Psi|}{\xx} + {{\cII}} \xx |\Psi| + {{\cIII}} |\Psi|^3 = \frac{\sigma}{|\Psi|^{3}}, \end{gather}
where the Madelung relation \eqref{35} shows that, in the present capillarity context $|\Psi|=\rho^{1/2}$. Thus, $\rho^{1/2}$ is governed by a hybrid Ermakov--Painlev\'e II equation while in terms of the density~$\rho$ it is seen that~\eqref{48} produces
\begin{gather} \label{49}
\deriv[2]{\rho}{\xx} = \frac{1}{2\rho} \left(\deriv{\rho}{\xx}\right)^2 - 2 {{\cIII}} \rho^2 - 2 {{\cII}} \xx \rho + \frac{2\sigma}{\rho},\end{gather}
which is equivalent to \Ptf\ \eqref{P34} (through a rescaling of the variables). This link between the Ermakov--Painlev\'e II equation \eqref{48} and \Ptf\ \eqref{P34} has been noted previously in the context of a Painlev\'e reduction of a classical Nernst--Planck electrodif\/fusion system in~\cite{pacr15}.
We remark that the special case of equation \eqref{49} with $c_3=0$ was considered by Gambier~\cite[pp.~27--28]{refGambier10}, who linearised the equation. Multiplying \eqref{49} with $c_3=0$ by $\rho$ and dif\/ferentiating gives
\begin{gather*} 
\deriv[3]{\rho}{\xx} = 4\xx\deriv{\rho}{\xx}+2\rho,
\end{gather*}
which has solution
\begin{gather} \label{49b}
 \rho(\xx)=C_1 \Ai^2(z)+C_2\Ai(z)\Bi(z)+C_3\Bi^2(z),\qquad z=-c_2^{1/3}\xx,
\end{gather}
with $C_1$, $C_2$ and $C_3$ constants. The solution $\rho(\xx)$ given by \eqref{49b} satisf\/ies \eqref{49} only if $c_3=0$, $\sigma=0$ and $4C_1C_2=C_3^2$.

In the sequel, it is convenient to proceed with
\begin{gather*} \label{50}
{{\cII}} = -\tfrac12 ,\qquad {{\cIII}} = -1 ,\qquad \sigma = -\tfrac14 \left( \a + \tfrac12 \right)^2, \end{gather*}
whence
\begin{gather} \label{51}
\bar{\ep} = \tfrac12\mathcal{C} \bar{\ga}^2 > 0 ,\qquad \la = \mathcal{C} \bar{\ga}^2 > 0 , \qquad
\left( \a + \tfrac12 \right)^2 = - \dfrac{4}{\la} \left[ \left( \dfrac{\mathcal{I}}{\bar{\ga}} \right)^2 + \mu \right],
\end{gather}
where the latter requires that $\mu<-\mathcal{C}\mathcal{I}^2/\la<0$. The Ermakov--Painlev\'e II equation \eqref{48} is then linked to \Ptf\ \eqref{P34}
via the relation $\rho=|\Psi|^2>0$.

The well-known connection, in turn, between \PII\ \eqref{P2} and \Ptf\ \eqref{P34} is readily derived via the Hamiltonian system
\begin{gather*} \label{53}
\deriv{q}{z} = \pderiv{\mathcal{H}_{\rm II}}{p} ,\qquad \deriv{p}{z} = - \pderiv{\mathcal{H}_{\rm II}}{q} , \end{gather*}
where the Hamiltonian $\mathcal{H}_{\rm II} ( p, q, z; \a )$ is given by
\begin{gather*} \label{54}
\mathcal{H}_{\rm II} ( p, q, z; \a ) = \tfrac12 p^2 - \big( q^2 + \tfrac12z \big) p - \big( \a + \tfrac12 \big) q, \end{gather*}
leading to the coupled pair of nonlinear equations
\begin{gather} \label{55}
\deriv{q}{z} = p - q^2 - \tfrac12z ,\qquad \deriv{p}{z} = 2 q p + \a + \tfrac12,\end{gather}
(see \cite{refJMi,refOkamotoPIIPIV}).
Elimination of $p$ and $q$ successively in \eqref{55} duly leads to the \PII\ \eqref{P2} and \Ptf\ \eqref{P34}. Thus, in the present capillarity context, the density distribution $\rho(\xx)$ is given by
\begin{gather*} \label{58}
\rho(\xx) = \deriv{w}{\xx} + w^2 + \tfrac12{\xx}, \end{gather*}
where $w(\xx)$ is governed by the \PII\ equation
\begin{gather*} \label{59}
\deriv[2]{w}{\xx} = 2w^3 + \xx w + \a.\end{gather*}

Here, the concern is necessarily restricted to solutions of \Ptf\ \eqref{P34} in regions in which $\rho$ is positive. Interestingly, the importance of positive solutions of \Ptf\ \eqref{P34} also arises naturally in the setting of two-ion electro-dif\/fusion. Thus, in the electrolytic context of \cite{lbjncrws10,ablbcr12}, the scaled electric f\/ield $Y$ was shown to be governed by the \PII\ equation
\begin{gather*} \label{60}
\deriv[2]{Y}{z}= 2Y^3 + z Y + \a, \end{gather*}
and associated ion concentrations by
\begin{gather*} \label{61}
p_\pm = \pm \deriv{Y}{z}+ Y^2 + \tfrac12z, \end{gather*}
with parameter
\begin{gather*} \label{62}
\a = \frac{1 - A_-/A_+}{2(1 + A_-/A_+)},\end{gather*}
and $A_\pm=-\Phi_\pm/D_\pm$, $\Phi_\pm$ being the f\/luxes of the ion concentrations and $D_\pm$ dif\/fusivity constants arising in the Einstein relation. Thus, it is seen that the ion concentrations, which are necessarily positive, are governed by \Ptf~\eqref{P34}. This positivity constraint was examined in detail in \cite{lbjncrws10} for exact solutions in terms of either Yablonskii-Vorob'ev polynomials or classical Airy functions as induced by the iterated action of the B\"acklund transformation of~\cite{nl71} for \PII~\eqref{P2}. The results apply `mutatis mutandis' in the present capillarity context.

\section{Iterated action of a B\"acklund transformation}

Here, the consequences of the following well-known B\"acklund transformation for \PII\ \eqref{P2} are applied in the present capillarity context.

\begin{Theorem}\label{thm41}
If $q_\a(z)=q(z;\a)$ is a solution of \PII\ \eqref{P2} with parameter $\a$, then
\begin{subequations}\label{64pm}\begin{gather} \label{64p}
q_{\a+1}(z) = -q_\a(z) - \frac{2\a + 1}{2q'_{\a}(z) + 2q^2_\a(z) + z}, \\
q_{\a-1}(z) = -q_\a(z) - \frac{2\a - 1}{2q'_{\a}(z) - 2q^2_\a(z) + z}, \label{64m}
\end{gather}\end{subequations}
are solution of \PII\ with respective parameters $\a+1$ and $\a-1$.
\end{Theorem}

\begin{proof} See Gambier \cite{refGambier10} and Lukashevich \cite{nl71}.\end{proof}

The iteration of the B\"acklund transformations \eqref{64pm} allows the generation of all known exact solutions of \PII\ \eqref{P2}.

We note that eliminating $q'_{\a}(z)$ in \eqref{64pm} yields the nonlinear dif\/ference equation
\begin{gather*}
\frac{\a+\tfrac12}{q_{\a+1}+q_{\a}}+\frac{\a-\tfrac12}{q_{\a}+q_{\a-1}}+2q_{\a}^2+z=0,
\end{gather*}
which is known as an alternative form of discrete Painlev\'e I \cite{refFGR}.

\begin{Theorem}\label{thm42}
If $q_\a=q(z;\a)$ and $p_\a=p(z;\a)$ are solutions of \PII\ \eqref{P2} and \Ptf\ \eqref{P34} with parameter $\a$ respectively, then
\begin{gather*}
q_{\a+1} = -q_\a - \frac{2\a + 1}{2p_{\a}}, \\
q_{\a-1} = -q_\a + \frac{2\a - 1}{2p_{\a}- 4q^2_\a + 2z}, \\
p_{\a+1} = -p_\a +\left(q_{\a}+ \frac{2\a + 1}{2p_{\a}}\right)^2+z, \\
p_{\a-1} = -p_\a + 2q^2_\a +z.
\end{gather*}
\end{Theorem}
\begin{proof}See Okamoto \cite{refOkamotoPIIPIV}; also \cite{refFW01}.\end{proof}

\subsection{Rational solutions}
The iterative action of the above B\"acklund transformation on the seed solution $q=0$ of \PII\ with $\a=0$ produces the subsequent sequence of rational solutions
\begin{gather} \label{65}
q_{n}(z) = \deriv{}{z} \ln \frac{Q_{n-1}(z)}{Q_{n}(z)},\qquad n \in\N, \end{gather}
corresponding to the Painlev\'e parameters $\a=n$, for $n\in\N$, where the $Q_{n}(z)$ are the \textit{Yablonskii--Vorob'ev polynomials} determined by the quadratic recurrence relations
\begin{gather} \label{66}
Q_{n+1} Q_{n-1} = z Q^2_{n} + 4\left\{\left(\deriv{Q_{n}}{z}\right)^2 - Q_{n} \deriv[2]{Q_{n}}{z}\right\} ,
\end{gather}
with $Q_{-1}(z) = Q_{0}(z) = 1$\ \cite{av65, ay59}; see also \cite{pc03+,pc03,pcem03,refKametaka86,refKanOch,refTaneda00}. The $Q_{n}(z)$ are monic polynomials of degree $\tfrac12n(n+1)$ with each term possessing the same degree modulo~$3$. Moreover, on use of the invariance under
\begin{gather*}q(z,\a)\rightarrow-q(z;-\a),\end{gather*}
it is seen that \PII\ \eqref{P2} also admits the associated class of rational solutions
\begin{gather} \label{67}
q_{-n}(z) = \deriv{}{z} \ln \frac{ Q_{n}(z)}{Q_{n-1}(z)} ,\qquad n \in\N, \end{gather}
corresponding to the Painlev\'e parameters $\a=-n$, for $n\in\N$. The rational solutions of~\Ptf\ \eqref{P34} are given by
\begin{gather*}
p_{n}(z) = \tfrac12z- 2\deriv[2]{}{z} \ln {Q_{n}(z)}\equiv \frac{Q_{n+1}(z)Q_{n-1}(z)}{2Q_{n}^2(z)},\qquad n \in\N, \end{gather*}
corresponding to the parameters $\a=n$, with $n\in\N$.

It is clear from the recurrence relation \eqref{66} that the $Q_{n}(z)$ are rational functions, though it is not obvious that they are polynomials since one is dividing by $Q_{n-1}(z)$ at every iteration. In fact it is somewhat remarkable that the $Q_{n}(z)$ are polynomials. Taneda \cite{refTaneda00}, used an algebraic method to prove that the functions $Q_{n}(z)$ def\/ined by~\eqref{66} are indeed polynomials, see also~\cite{refFOU}. The Yablonskii--Vorob'ev polynomials $Q_{n}(z)$ can also be expressed as determinants, see \cite{pcem03,refKMi,refKO96}. Clarkson and Mansf\/ield~\cite{pcem03} investigated the locations of the roots of the Yablonskii--Vorob'ev polynomials in the complex plane and showed that these roots have a very regular, approximately triangular structure; the term ``approximate" is used since the patterns are not exact triangles as the roots lie on arcs rather than straight lines. Recently Bertola and Bothner~\cite{refBB15} and Buckingham and Miller~\cite{refBM14,refBM15} have studied the Yablonskii--Vorob'ev polynomials $Q_{n}(z)$ in the limit as $n\to\infty$ and shown that the roots lie in a~``triangular region'' with elliptic sides which meet with interior angle~$\tfrac25\pi$, suggesting a limit to a solution of Painlev\'e~I~(\PI). Indeed Buckingham and Miller \cite{refBM15} show that in the limit as $n\to\infty$, the rational solution $q_{n}(z)$ of \PII\ tends to the \textit{tritronqu\'ee solution} of \PI\ due to Boutroux \cite{boutroux}, which no poles (of large modulus) except in one sector of angle $\tfrac25\pi$\ (see also \cite{refJK}).

In the sequel, attention is restricted to the case with invariants $\mathcal{I}=\mathcal{I}_{n}\neq0$ so that similarity reduction of the capillarity system via \eqref{2-3} leads to consideration of a Ermakov--Painlev\'e II equation in the amplitude $|\Psi|$ and, in turn, density $\rho$ as determined by \Ptf\ \eqref{49}. Thus, the density distribution $\rho=\rho_+$ associated with the class of exact solutions \eqref{65} of \PII\ \eqref{P2} in terms of Yablonskii--Vorob'ev polynomials with $z=\xx$ may be shown to be given by rational expressions derived via \eqref{48} to adopt the form (see \cite{pc03+,pc03,pcem03})
\begin{gather} \label{68}
\rho_+(\xx;n) = \frac{Q_{n+1}(\xx) Q_{n-1}(\xx)}{2Q^2_{n}(\xx)} ,\qquad \rho_+(\xx;0) = \tfrac12{\xx},
 \end{gather}
as subsequently employed in the two-ion electolytic boundary value problems investigated in~\cite{lbjncrws10}. Therein, the positivity of members of the class of rational solutions \eqref{68} of \Ptf\ on appropriate regions has been delimited. In the present capillarity context, in such regions where $\rho=\rho_+$ is positive, the class of solutions \eqref{68} is associated with wave packet representations~\eqref{2-3} with
\begin{subequations} \label{69}\begin{gather}
\phi(\xi) = \pm \frac{V(\xi)}{\sqrt{1+V^2(\xi)}} \sqrt{\dfrac{Q_{n+1}(\xi) Q_{n-1}(\xi)}{2Q^2_{n}(\xi)}} ,\\
\psi(\xi) = \pm \frac{1}{\sqrt{1+V^2(\xi)}}\sqrt{\dfrac{Q_{n+1}(\xi) Q_{n-1}(\xi)}{2Q^2_{n}(\xi)}},
\end{gather}\end{subequations}
where $V(\xi)$ is given by, in view of \eqref{27},
\begin{gather}
\bar{\ga}^2 \tan^{-1} V(\xi) = \left( \bar{\a} - \frac{\bar{\delta}}{\bar{\ep}} \right) \xi + 2\mathcal{I}_{n} \int^{\xi} \frac{Q^2_{n}(\xii)}{Q_{n+1}(\xii) Q_{n-1}(\xii)}\,\d \xii \nonumber\\
\hphantom{\bar{\ga}^2 \tan^{-1} V(\xi)}{} = \left( \bar{\a} - \frac{\bar{\delta}}{\bar{\ep}} \right) \xi + \frac{2\mathcal{I}_{n}}{2n+1}\ln \frac{Q_{n+1}(\xi)}{Q_{n-1}(\xi)},
\label{70}\end{gather}
since
\begin{gather} \label{70a}
\int^\xi \frac{Q^2_{n}(\xii)}{Q_{n+1}(\xii) Q_{n-1}(\xii)}\,\d \xii = \frac{1}{2n+1}\ln \frac{Q_{n+1}(\xi)}{Q_{n-1}(\xi)}.
\end{gather}
This result \eqref{70a} follows since the Yablonskii--Vorob'ev polynomials $Q_{n}(\xi)$ satisfy the bilinear relation
\begin{gather*} 
\deriv{Q_{n+1}}{\xi}Q_{n-1}-\deriv{Q_{n-1}}{\xi}Q_{n+1}=(2n+1)Q_{n}^2,
\end{gather*}
which is proved in \cite{refFOU,refTaneda00} (see also \cite{refKanOch}).
Likewise, there are associated classes of wave packet representations corresponding to the rational solution $q_{-n}(z)$ given by \eqref{67} which are determined by the relations \eqref{69}, \eqref{70} but with the transposition $n\leftrightarrow n-1$.

\subsection{Airy-type solutions}
The iterated action of the B\"acklund transformations \eqref{64pm} may, in addition, be used to ge\-nerate exact solutions of \PII\ \eqref{P2} with parameters $\a=\pm\tfrac12, \pm\tfrac32, \ldots$ in terms of classical Airy functions~\cite{pc16}. Thus, in particular, if $\a=\tfrac12$ then \PII\ \eqref{P2} admits the exact solution
\begin{gather} \label{71}
q\big(z; \tfrac12\big) = - \deriv{}{z} \ln \ph(z),\end{gather}
where $\ph(z)$ is governed by the classical Airy function
\begin{gather} \label{72}
\deriv[2]{\ph}{z} + \tfrac12 z \ph = 0.\end{gather}
Iteration of the B\"acklund transformation \eqref{64p} with the Airy-type seed solution \eqref{71} generates as inf\/inite sequence of exact solutions
\begin{gather} \label{73}
q\big(z;n-\tfrac12\big) = \deriv{}{z} \ln \frac{u_{n-1}(z)}{u_{n}(z)} ,\qquad n \in \N, \end{gather}
where the sequence $\{u_\ell(z)\}$, for $\ell\geq0$, is determined by the recurrence relation (Toda equation)
\begin{subequations}\label{74-75} \begin{gather} \label{74}
u_{n+1} u_{n-1} = 4 \left\{\left( \deriv{u_{n}}{z}\right)^2 - u_{n} \deriv[2]{u_{n}}{z}\right\}, \end{gather}
with initial values
\begin{gather} \label{75}
u_{0}(z) = 1 ,\qquad u_1(z) = \ph(z).\end{gather}\end{subequations}
The $u_{n}(z)$, for $n\geq2$, are homogeneous polynomials of degree $n$ in $\ph(z)$ and $\ph'(z)$, i.e., have the form
\begin{gather*}
u_{n}(z) = \sum_{j=0}^{n}a_{n,j} \ph^j\left(\deriv{\ph}{z}\right)^{n-j},
\end{gather*}
where $a_{n,j}(z)$ are polynomials in $z$
and $\ph(z)$ is the solution of \eqref{72} given by
\begin{gather} \label{76}
\ph(z) = \cos(\th) \Ai \big({-}2^{-1/3} z\big) + \sin(\th) \Bi \big({-}2^{-1/3} z\big),
\end{gather}
with $\Ai(\zz)$ and $\Bi(\zz)$ the Airy functions and $\th$ an arbitrary constant.

The analogous Airy-type solutions of \Ptf\ \eqref{P34} are given by
\begin{gather*} 
p\big(z;n-\tfrac12\big) = - 2\deriv[2]{}{z} \ln {u_{n}} \equiv \frac{u_{n-1}u_{n+1}}{2u_{n}^2},\qquad n \in\N, \end{gather*}
for the parameter $\a=n-\tfrac12$.

The Airy-type solutions of \PII\ \eqref{P2} and \Ptf\ \eqref{P34} can also be expressed in terms of determinants, as described in the following theorem.
\begin{Theorem}\label{thm43}Let $\tau_{n}(z)$ be the Hankel $n\times n$ determinant
\begin{gather*}
\tau_{n}(z) = \left[\deriv[j+k]{}{z}\varphi(z) \right]_{j,k=0}^{n-1},\qquad n\geq1,
\end{gather*}
with $\varphi(z)$ given by \eqref{76} and $\tau_0(z)=1$, then for $n\geq1$,
\begin{gather} 
q\big(z;n-\tfrac12\big)= \deriv{}{z}\ln\frac{\tau_{n-1}(z)}{\tau_{n}(z)},\qquad
\label{pnAiry} p\big(z;n-\tfrac12\big) = - 2\deriv[2]{}{z} \ln {\tau_{n}(z)},
\end{gather}
respectively satisfy \PII\ \eqref{P2} and \Ptf\ \eqref{P34} with $\a=n-\tfrac12$.\end{Theorem}

\begin{proof}See Flaschka and Newell \cite{refFN}, Okamoto \cite{refOkamotoPIIPIV}; also \cite{pc16,refFW01}.\end{proof}

We remark that recently it was shown that Airy-type solutions of \PII\ \eqref{P2} and \Ptf~\eqref{P34} which depend only on the Airy function $\Ai(\zz)$ have a completely dif\/ferent structure to those which involve a linear combination of the Airy functions $\Ai(\zz)$ and $\Bi(\zz)$, see~\cite{pc16}. In particular, for $n\in2\N$ the solution \eqref{pnAiry} of \Ptf\ \eqref{P34} has no poles on the real axis when $\varphi(z)=\Ai(-2^{-1/3} z)$ and decays algebraically as $z\to\pm\infty$. This special solution arose in a study of the double scaling limit of unitary random matrix ensembles by Its, Kuijlaars, and \"{O}stensson \cite{refIKO08,refIKO09}, who identify the solution as a \textit{tronqu\'{e}e} solution of \Ptf, i.e., has no poles in a sector of the complex plane (see also~\cite{pc16}).

The iterative application of the B\"acklund transformations \eqref{64pm} to generate Airy-type solutions of \PII\ \eqref{P2} has been used to solve boundary value problems associated with the classical Nernst--Planck system for two-ion electro-dif\/fusion \cite{lbjncrws10,crabws99}. The repeated action of the B\"acklund transformations in this electrolytic setting has been recently associated with quantised f\/luxes of ionic species in \cite{ablbcr12}. It is remarked that the Painlev\'e structure underlying a multi-ion electrodif\/fusion model has been systematically investigated in~\cite{rccrws07}.

In the present capillarity context, the density distributions associated with the class of exact solutions of \PII\ \eqref{P2} determined by the relations \eqref{73}, \eqref{74-75} are given by \cite{lbjncrws10}
\begin{gather*} \label{77}
\rho\big(\xi;n - \tfrac12 \big) = \frac{u_{n-1}(\xi) u_{n+1}(\xi)}{2u^2_{n}(\xi)} ,\qquad \rho \big(\xi;-\tfrac12\big) = 0,
\end{gather*}
where $u_{n}(z)$ is given by \eqref{74-75}. The physical requirement of positivity of these solutions corresponding to the specialisation $\ph(z)= \Ai(-2^{-1/3}z)$ in \eqref{76}, i.e., when $\th=0$, has been examined in detail in the electrolytic context of~\cite{lbjncrws10}. The implications of the results carry over to the present capillarity study.

In general, with regard to the velocity magnitude $v=|\bf{v}|$, alignment of the Madelung transformation \eqref{35} with \eqref{2-3} produces the velocity potential relation
\begin{gather} \label{79}
\Phi = 2\tan^{-1} \left( \frac{\psi + \phi \tan \eta}{\phi - \psi \tan \eta} \right) = 2 \left[ \tan^{-1} \left( \frac{\psi}{\phi} \right) + \eta \right],\end{gather}
on use of the identity
\begin{gather*} 
\tan^{-1} \left( \frac{x + y}{1 - x y} \right) \equiv \tan^{-1} x + \tan^{-1} y.\end{gather*}
Integration of the invariant relation \eqref{9} where in the present context $c=0$, yields
\begin{gather*}
-\bar{\ga}^2 \tan^{-1} \left( \frac{\psi}{\phi} \right)
= \left( \bar{\a} - \frac{\bar{\delta}}{\bar{\ep}} \right) \xi +2\mathcal{I} \int^\xi \frac{u_{n}^2(\xii)}{u_{n+1}(\xii)u_{n-1}(\xii)}\,\d \xii
= \left( \bar{\a} - \frac{\bar{\delta}}{\bar{\ep}} \right) \xi + \frac{\mathcal{I}}{n} \ln\frac{u_{n+1}(\xi)}{u_{n-1}(\xi)},
\end{gather*}
since
\begin{gather} \label{81a}
\int^\xi \frac{u^2_{n}(\xii)}{u_{n+1}(\xii) u_{n-1}(\xii)}\,\d \xii = \frac{1}{2n}\ln \frac{u_{n+1}(\xi)}{u_{n-1}(\xi)}.
\end{gather}
The result \eqref{81a} holds as the $u_{n}(\xi)$ satisfy the bilinear relation
\begin{gather*} 
\deriv{u_{n+1}}{\xi}u_{n-1}-\deriv{u_{n-1}}{\xi}u_{n+1}=2nu_{n}^2,
\end{gather*}
which follows from Theorem \ref{thm42}.
Consequently \eqref{79} yields
\begin{gather*}
|\mathbf{v}| = -\frac{2}{\bar{\ga}^2} \left[ \bar{\ga} \pderiv{}{\xi} + \left( \bar{\ep} \bar{\ga} t + \bar{\la} \right) \pderiv{}{\eta} \right] \left[ \left( \bar{\a} - \frac{\bar{\delta}}{\bar{\ep}} \right) \xi +{\mathcal{I}}\int \frac{1}{|\Psi|^2}\,\d \xi + \eta \right], \end{gather*}
whence,
\begin{gather*} \label{82}
v = - \frac{2\mathcal{I}}{\bar{\ga} \rho (\xi)} + 2 \bar{\ep} \bar{\ga} t - \frac{\bar{\a}}{\bar{\ga}},\end{gather*}
on use of the relation \eqref{13}. Insertion of the latter expression into the momentum equation of the capillarity system, namely
\begin{gather*} \label{83}
v_t + vv_x + \left[ \frac{\delta}{\delta \rho} (\rho \mathcal{E}) \right]_x = 0, \end{gather*}
on integration, produces the Bernoulli integral
\begin{gather*} \label{84}
\frac{2\mathcal{I}^2}{\bar{\ga}^2 \rho^2} + 2 \bar{\ep} \bar{\ga} x + \frac{\delta}{\delta \rho} (\rho \mathcal{E}) = \mathcal{B}(t).\end{gather*}
Here \eqref{43} shows that
\begin{gather*} \label{85}
\rho \mathcal{E} = \mathcal{C} \frac{ \bar{\ga}^2}{2\rho} \rho^2_x + \la \rho^2 + \frac{2\mu}{\rho} + 2\nu \rho + \tau , \end{gather*}
whence,
\begin{gather*}
\frac{\delta}{\delta \rho} (\rho \mathcal{E}) = \pderiv{}{\rho} (\rho \mathcal{E}) - \mathcal{C} \pderiv{}{x} \left( \frac{\rho_x}{\rho} \right) = 2 \la \rho + \frac{2\mu}{\rho^2} + 2\nu + \mathcal{C} \left(\frac{\rho^2_x}{2\rho^2} - \frac{\rho_{xx}}{\rho} \right), \end{gather*}
in which $\rho(x,t)=R(\xi)$ with $\xi=\bar{\a}t+\bar{\b}t^2+\bar{\ga}x$ so that $\rho_x=\bar{\ga}R'(\xi)$. Thus,
\begin{gather*} \label{86}
\frac{2\mathcal{I}^2}{\bar{\ga}^2 R^2} + 2 \bar{\ep} \bar{\ga} x + 2 \la R + \frac{2\mu}{R^2} + 2\nu + \mathcal{C} \bar{\ga}^2 \left[ \frac{1}{2R^2} \left(\deriv{R}{\xi}\right)^2- \frac{1}{R} \deriv[2]{R}{\xi}\right] = \mathcal{B}(t), \end{gather*}
and with $\mathcal{B}(t)=-2\bar{\ep}(\bar{\a}t+\bar{\b}t^2)$, an equation equivalent to \Ptf\ \eqref{P34} for the density $\rho(\xx)$ is retrieved, namely
\begin{gather*} 
\deriv[2]{\rho}{\xx} = \dfrac{1}{2\rho}\left(\deriv{\rho}{\xx}\right)^2 + \dfrac{2 \la\rho^2}{\mathcal{C} \bar{\ga}^2} + \dfrac{2 \bar{\ep} \xx\rho}{ \mathcal{C} \bar{\ga}^2} + \dfrac{2(\mathcal{I}^2+\gamma^2\mu)}{\mathcal{C} \bar{\ga}^4 \rho}, \qquad \xx=\xi+\dfrac{{{\cI}}+{{\cV}}}{{{\cII}}} ,
\end{gather*}
which aligns with \eqref{49} in view of the relations \eqref{44}, \eqref{47} and \eqref{51}. Thus, it is seen that, remarkably, in the present context \Ptf\ is associated with the density $\rho$ corresponding to a~symmetry reduction of the Bernoulli integral of motion of the capillarity system.

\section{Invariance under a reciprocal transformation}

The application of reciprocal-type transformations to $(1+1)$-dimensional nonlinear physical systems has its origin in the isolation of novel invariance properties in gasdynamics and magnetogasdynamics~\cite{cr68,cr69}. They have been subsequently applied to both obtain analytic solution to moving boundary problems of Stefan-type \cite{cr86} and to link integrable systems of modern soliton theory (see, e.g.,~\cite{addhah02, wocr93} together with~\cite{crws02} and work cited therein). In $3+1$ dimensions, reciprocal-type transformations have been shown to have application in discontinuity wave propagation theory~\cite{adurcr92}.

In the present capillarity context, the $(1+1)$-dimensional version of the system \eqref{29-30} with $\Pi=0$, namely
\begin{gather*} 
\rho_t + (\rho v)_x = 0 , \qquad v_t + vv_x + \left[ \dfrac{\delta}{\delta \rho} (\rho \mathcal{E}) \right]_x = 0,
\end{gather*}
was recently shown in \cite{crws14} to be invariant under the one-parameter $(\chi)$ class of reciprocal-type transformations
\begin{gather*} 
\rho^* = \dfrac{\rho}{1 + \chi \rho} ,\qquad q^* = q ,\qquad \mathcal{E}^* (\mathcal{A}^*, \rho^*) = \mathcal{E} (\mathcal{A}, \rho) ,
\qquad \mathcal{A}^* = \dfrac{\mathcal{A}}{(1 + \chi \rho)^6} , \\
\d x^* = (1 + \chi \rho)\,\d x - \chi \rho q\,\d t ,\qquad\d t^* =\,\d t , \qquad 0 < | 1 + \chi \rho | < \infty,
\end{gather*}
where $\mathcal{A}=\tfrac12\rho^2_x$. A direct corollary of this result is that the $(1+1)$-dimensional Korteweg-type capillarity system
\begin{subequations} \label{90}
\begin{gather}
\rho_t + (\rho v)_x = 0 , \\
v_t + vv_x + \left( - \kappa(\rho) \rho_{xx} - \tfrac12 \deriv{\kappa}{\rho} \rho^2_x + \deriv{\mathcal{R}}{\rho} \right)_x = 0,
\end{gather} \end{subequations}
is invariant under the one-parameter class of reciprocal transformations
\begin{gather*} 
\d x^* = (1 + \chi \rho)\,\d x - \chi \rho q\,\d t ,\qquad\,\d t^* =\d t,\qquad \rho^* = \dfrac{\rho}{1 + \chi \rho} ,\qquad q^* = q,
\end{gather*}
augmented by the relations
\begin{gather*} \label{92}
\kappa^* = (1 + \chi \rho)^5 \kappa ,\qquad \mathcal{R}^* = \frac{\mathcal{R}}{1 + \chi \rho}.\end{gather*}
This invariant transformation may be applied to seed solutions of the capillarity system \eqref{90} as previously determined in terms of Yablonskii--Vorob'ev polynomials or classical Airy functions to construct extended $\chi$-dependent classes of exact solutions valid for model energy laws $\mathcal{E}^*$ with $\chi$-deformed K\'{a}rm\'{a}n--Tsien capillarity relation
\begin{gather*} \label{93}
\kappa^* = \frac{\mathcal{C}}{\rho^* (1 - \chi \rho^*)^4} , \end{gather*}
together with
\begin{gather*} \label{94}
\mathcal{R}^* = \frac{\la \rho^{*2}}{1 - \chi \rho^*} - 2 \mu \frac{(1 - \chi \rho^*)^2}{\rho^*} + 2 \nu \rho^* + \tau (1 - \chi \rho^*).\end{gather*}
The original $\kappa(\rho)$, $\mathcal{R}(\rho)$ associated with the \Ptf\ reduction are retrieved in the limit $\chi\rightarrow0$.

\section[General perspectives on model laws in continuum mechanics and solitonic connections: conclusion]{General perspectives on model laws in continuum mechanics\\ and solitonic connections: conclusion}

Here, model energy laws have been isolated for which a Korteweg-type capillarity system admits symmetry reduction to a integrable hybrid Ermakov--Painlev\'e equation. A B\"{a}cklund and reciprocal transformation have been used in turn to generate novel classes of exact solutions and to extend the range of the reduction. The derivation of multi-parameter model constitutive laws for which systems in nonlinear continuum mechanics become analytically tractable via the application of B\"{a}cklund or reciprocal transformations has an extensive literature. Thus, in gasdynamics, Loewner \cite{cl50,cl52} applied matrix B\"acklund transformations to construct model constitutive laws for which the classical hodograph equations may be systematically reduced to appropriate tractable canonical forms in subsonic, transonic and supersonic f\/low r\'{e}gimes. The celebrated K\'{a}rm\'{a}n--Tsien two-parameter pressure-density model law of \cite{ht39} as extensively applied in subsonic gasdynamics, arises as a particular reduction. The B\"acklund transformations as introduced in the model gas law context of \cite{cl52}, suitably interpreted and extended, remarkably, turn out to have application in $(2+1)$-dimensional soliton theory \cite{bkcr91,bkcr93}. In nonlinear elastodynamics, model multi-parameter stress-strain laws were constructed in \cite{hcev74} which allow the analytic treatment of aspects of shock-less pulse propagation in bounded nonlinear elastic media. Comparison of experimental stress-strain relations with such model laws was investigated, in particular, for the dynamic compression of saturated soil, dry sand and clay silt. A~B\"acklund transformation may be introduced at the level of the stress-strain laws for the uniaxial Lagrangian elastodynamic system treated in~\cite{hcev74}. The single action of this B\"acklund transformation to the classical Hooke's law generates the multi-parameter class of $(T, e)$ laws applied extensively therein. Moreover application of a nonlinear superposition principle associated with the B\"acklund transformation permits the construction of more general model nonlinear $(T, e)$-laws for which the $(1+1)$-dimensional elastodynamic system may be iteratively reduced to that associated with the canonical Hooke's law. There is again a remarkable solitonic connection in that the nonlinear superposition principle acting on the $(T, e)$-laws turns out to be nothing but the permutability theorem for the potential Korteweg--de~Vries hierarchy (see, e.g.,~\cite{crws02}).

In nonlinear elastostatics, model stress-deformation laws have been introduced by Neuber in~\cite{hn58,hn61} in connection with the problem of determining the stress-distribution in shear-strained isotropic prismatical bodies. Loewner-type B\"acklund transformations were applied in~\cite{dccr75} to the Neuber elastostatic system to solve a class of indentation boundary value problems for both Neuber--Sokolovsky and power law model stress-deformation relations. The application of model $\bf{B}$-$\bf{H}$ and $\bf{D}$-$\bf{E}$ constitutive laws in the analysis of the propagation of plane polarised electromagnetic waves through nonlinear dielectric media has been described in~\cite{jkrv75}. In general terms, it was shown in \cite{wscr98} that model constitutive laws as constructed via the B\"{a}cklund approach introduced by Loewner, corresponds to solitonic solutions generated by a Darboux-type transformation.

With regard to reciprocal transformations and their role in the construction of model constitutive laws one may cite, in particular, the investigation of Storm in~\cite{ms51} concerning heat conduction in simple monatomic metals. Therein, a class of model $(c_p(T),k(T))$ temperature $T$-dependent laws was introduced for which a $(1+1)$-dimensional nonlinear heat conduction equation may be reduced via a reciprocal transformation to the classical linear heat equation. The applicability of these model laws was justif\/ied in \cite{ms51} for appropriate specif\/ic heat $c_p(T)$ and thermal conductivity $k(T)$ over wide temperature ranges for such materials as aluminium, silver, sodium, cadium, zinc, copper and lead. This kind of reduction via a reciprocal transformation may be extended to hyperbolic systems that correspond to multi-parameter model laws that arise in Cattaneo-type conduction and nonlinear visco-elasticity (see, e.g.,~\cite{bsev84}).

The preceding attests to the importance and wide range of physical applications of model constitutive laws in nonlinear continuum mechanics together with intriguing solitonic connections. The six classical Painlev\'e equations arise in a wide range of physical applications and play a fundamental role in modern soliton theory (see, e.g., \cite{pc05,rc99,fikn06,refGLS}). In the present work, model multi-parameter specif\/ic energy laws have been isolated which allow symmetry reduction of a Korteweg capillarity system to consideration of a hybrid Ermakov--Painlev\'e II equation and thereby to the linked integrable \Ptf\ equation. In conclusion, it is remarked that the Ermakov--Painlev\'e II symmetry reduction presented here is also valid for a superf\/luidity model system involving a de Broglie Bohm potential as set down in \cite{prnb01}.

\pdfbookmark[1]{References}{ref}
\LastPageEnding

\end{document}